\documentclass[conference]{IEEEtran}
\IEEEoverridecommandlockouts
\usepackage[T1]{fontenc}
\usepackage{times}
\usepackage{cite}
\usepackage{amsmath,amssymb,amsfonts,amsthm,gensymb}
\usepackage{algorithmic}
\usepackage{graphicx}
\usepackage{textcomp}
\usepackage{xcolor}
\usepackage{comment}
\usepackage{subfig}

\usepackage{acro}

\newtheorem{lemma}{Lemma}

\DeclareAcronym{MMSE}{
  short = MMSE,
  long  = minimum mean-square-error,
}
\DeclareAcronym{MSE}{
  short = MSE,
  long  = mean-square-error,
}
\DeclareAcronym{UATF}{
  short = UATF,
  long  = use-and-then-forget,
}\DeclareAcronym{UL}{
  short = UL,
  long  = uplink,
}
\DeclareAcronym{MIMO}{
  short = MIMO,
  long  = multiple-input multiple-output,
}
\DeclareAcronym{MR}{
  short = MR,
  long  = maximum-ratio,
}
\DeclareAcronym{BS}{
  short = BS,
  long  = base station,
}
\DeclareAcronym{SINR}{
short = SINR,
long=signal-to-interference-plus-noise ratio
}

\DeclareAcronym{PAS}{
short = PAS,
long=power angular spectra
}

\DeclareAcronym{PSD}{
short = PSD,
long=positive semidefinite,
}

\DeclareAcronym{ASD}{
short = ASD,
long=angular standard deviation,
}

\DeclareAcronym{SE}{
short=SE,
long=spectral efficiency
}

\DeclareAcronym{CSI}{
short=CSI,
long=channel state information
}

\DeclareAcronym{SNR}{
short=SNR,
long=signal-to-noise ratio
}

\def\BibTeX{{\rm B\kern-.05em{\sc i\kern-.025em b}\kern-.08em
    T\kern-.1667em\lower.7ex\hbox{E}\kern-.125emX}}
\begin{document}

\title{Exploiting Pilot Contamination to Improve UL Sum Rate under Maximum Ratio Combining\\
\thanks{This work was supported by the IS24-0190 grant from the Swedish Foundation for Strategic Research.}
}

\author{\IEEEauthorblockN{ Emanuel Wennemo, Ki Won Sung, Emil Björnson}
\IEEEauthorblockA{Department of Communication Systems, KTH Royal Institute of Technology, Stockholm, Sweden \\
Email: \{wennemo, sungkw, emilbjo\}@kth.se}}

\maketitle

\begin{abstract}
In this work, we show that pilot contamination can be exploited to jointly increase the uplink sum rate and reduce pilot overhead. By leveraging the effect of pilot contamination on \ac{MMSE} channel estimates, we demonstrate that allowing users to share or use correlated pilots can, in certain regimes, yield better sum rate than with mutually orthogonal pilots or perfect \ac{CSI} under \ac{MR} combining. To establish this result, we derive the \ac{UATF} capacity bound for an arbitrary pilot design under correlated Rayleigh fading with \ac{MR} combining. We show that, in the presence of pilot contamination, the key mechanism is the directional suppression of the channel estimate towards nearby interfering angular regions, enabling interference reduction. Numerical results confirm this behavior, demonstrating the sum-rate gain and showing that it grows with the number of antennas.   
\end{abstract}
\acresetall

\section{Introduction}
Pilot contamination arises when users do not have mutually orthogonal pilots, leading to correlation between users' channel estimates \cite{Cell_Free_Book}. Traditionally, this has been studied in Multi-Cell and Cell-Free massive \ac{MIMO} systems \cite{Cell_Free_PC} as an adversary to be mitigated \cite{Survey}, with the single-cell setting explicitly mentioned as an area where the effect is lost \cite{Marzetta_pilotcon}. However, with the emergence of the upper mid-band (7-24 GHz) as a likely candidate for 6G, where higher carrier frequencies lead to larger \ac{MIMO} configurations but a shrinking coherence time and reduced pilot budget \cite{FR3}, we cannot dismiss the single-cell setting entirely. This motivates examining the effect of pilot contamination in this setting, and questioning whether it is truly an adversary.

Extensive work has been done to mitigate pilot contamination, some by assigning users to reuse a limited set of orthogonal pilots \cite{Hungarian,SGPS}, while others consider non-orthogonal pilot design to turn the combinatorial problem into the continuous one \cite{Non_orth}. Across these works, the common objective is to minimize the effect of pilot contamination, either explicitly through \ac{MSE} minimization, or implicitly through user separation. However, in this paper we show that, under \ac{MR} combining, not even a perfect channel estimate is necessarily optimal for the \ac{UL} sum rate, since the channel estimate directly determines the combining vector.

The observation that minimizing MSE is not necessarily optimal under \ac{MR} combining relates to the design of bilinear equalizers \cite{bilinear}, where the \ac{MMSE} channel estimation matrix is replaced by an arbitrary statistic-dependent filter, optimized to maximize the \ac{SINR} in the \ac{UATF} capacity bound \cite{Cell_Free_Book}. Rather than optimizing the filter matrix freely, one could retain the structure of the \ac{MMSE} channel estimation matrix and shape it indirectly through pilot contamination, thereby remaining within the conventional estimation framework while actively reducing pilot overhead. 

In this work, we pursue this approach by exploiting the effect of pilot contamination on \ac{MMSE} channel estimates. Using a non-orthogonal pilot design framework, we highlight how pilot contamination can be leveraged to suppress interference and achieve a larger \ac{UL} sum rate compared to users having mutually orthogonal pilots or even perfect \ac{CSI}, in certain regimes under \ac{MR} and normalized \ac{MR} combining. To enable this study, we derive the \ac{UATF} capacity bound for an arbitrary pilot design under correlated Rayleigh fading with \ac{MR} combining, a result that, to the best of our knowledge, has not been reported in the prior literature. With this, we show that pilot contamination is not necessarily an adversary and can be exploited to jointly improve the sum rate and reduce pilot overhead.

\section{System Model}
We consider the \ac{UL} of a single-cell massive \ac{MIMO} system with $K$ user equipments and a single \ac{BS}. The \ac{BS} consists of $M$ antennas arranged as a horizontal uniform linear array with antenna spacing $d$
in wavelengths, while each user is equipped with a single antenna and transmits simultaneously. Under the assumption of block-flat fading, we model the channel to user $k$ in an arbitrary block as $\mathbf h_k \sim \mathcal{N}_\mathbb{C}(\mathbf 0_M,\mathbf R_k)$ \cite{Corr_fading},
where $\mathcal N_\mathbb C$ denotes the circularly symmetric complex Gaussian distribution and the element $(l,m)$ of the spatial correlation matrix $\mathbf R_k$ is given by
\begin{equation}
[\mathbf{R}_k]_{l,m}=\beta_k\int_{-\infty}^\infty e^{j2\pi (l-m)d  \sin(\overline\varphi_k)}f(\overline \varphi_k)d \overline \varphi_k.
\label{Spat_cor}
\end{equation}
This model assumes a large number of multipath components, where $f(\overline \varphi_k)$ is the probability density function of the angle of arrival, while $\beta_k$ is the total average gain. The second-order statistics, $\mathbf R_k$, are assumed known, justified by their slow variation over time and frequency \cite{SGPS}.

Given that each user $k$ transmits the pilot signal $\boldsymbol  \phi_k\in \mathbb C^{ \tau \times 1}$, with squared norm $\tau$, length $\tau$, and transmit power $\eta_k$, the \ac{BS} receives the signal $\mathbf Y^{p}\in \mathbb C^{M\times \tau}$ expressed as \begin{equation}
\mathbf Y^{p}=\sum_{j=1}^K \sqrt{\eta_j} \mathbf h_j \boldsymbol \phi_j^T+\mathbf N,
\end{equation}
where $\mathbf N\in \mathbb C^{M\times \tau}$ is the receiver noise with each element modeled as $\mathcal N_\mathbb C(0,\sigma^2_{\mathrm{ul}})$ and statistically independent. To estimate $\mathbf h_k$, we multiply the received signal $\mathbf Y^{p}$ with $\boldsymbol \phi^*_k/\sqrt\tau$ and obtain
\begin{equation}
\mathbf y_k^{p}=\sum_{j=1}^K\sqrt{\eta_j}\mathbf h_j  \frac{\boldsymbol\phi_j^T \boldsymbol \phi_k^*}{\sqrt{\tau}}+\frac{\mathbf N \boldsymbol \phi_k^*}{\sqrt \tau},
\label{Projector}
\end{equation}
where we note that $\mathbf N \boldsymbol \phi_k^*/\sqrt \tau \sim \mathcal N_\mathbb{C}(\mathbf 0_M,\sigma^2_{\mathrm{ul}}\mathbf I_M)$. Defining $\boldsymbol{\Gamma}=\boldsymbol \Phi^H\boldsymbol \Phi/\tau$,
 with $\boldsymbol{\Phi}=\begin{bmatrix}
    \boldsymbol\phi_1 & \boldsymbol\phi_2&  ... &\boldsymbol\phi_K
\end{bmatrix}\in \mathbb C^{\tau\times K},$
let us rewrite \eqref{Projector} without restricting the pilot sequences to any specific design:
\begin{equation}
    \mathbf y_{k}^{p}=\underbrace {\sqrt{\eta_k\tau}\mathbf h_{k}}_{\text{Desired part}}+\underbrace{\sum_{j=1,j\neq k}^{K}\sqrt{\eta_j\tau}\mathbf h_{j}\boldsymbol \Gamma_{j,k}}_{\text{Interference}}+\underbrace{\mathbf n_{k}}_{\text{Noise}},
    \label{y_pilot}
\end{equation}
where $\mathbf n_k\sim \mathcal N_{\mathbb C}(\mathbf0_M,\sigma^2_{\mathrm{ul}}\mathbf I_M)$ and $\boldsymbol \Gamma_{j,k}$ is the normalized pilot correlation between user $j$ and $k$, obtained as the $(j,k)$-th entry of $\boldsymbol{\Gamma}$. The \ac{MMSE} channel estimate based on $\mathbf y_{k}^{p}$ is then given by\footnote{When orthogonal pilots are not assumed, $\mathbf y_{k}^{p}$ is no longer a sufficient statistic, and the \ac{MMSE} estimate can instead be based on $\mathbf Y^{p}$, with the drawback of higher complexity. We restrict the channel estimate to $\mathbf y_{k}^{p}$, as using $\mathbf Y^{p}$ only leads to a marginal gain within the considered framework.} 
\begin{equation}
        \hat{\mathbf h}_k=\sqrt{\eta_k\tau}\mathbf R_k\mathbf \Psi^{-1}_k\mathbf y_{k}^{p},
\end{equation}
where $\mathbf \Psi_k=\mathbb E[\mathbf y_{k}^p(\mathbf y_{k}^p)^H]=\sum_{j=1}^{K}\eta_j\tau\mathbf R_{j}|\boldsymbol \Gamma_{j,k}|^2+\sigma_{\mathrm{ul}}^2\mathbf{I}_{M}$. 

\section{Uplink SINR Analysis}
To study and later evaluate how pilot contamination affects performance, we now derive the \ac{UATF} capacity bound for an arbitrary pilot design. An achievable rate [bit/s/Hz] for user $k$, based on the \ac{UATF} bound, can be expressed as \cite{Cell_Free_Book}
\begin{equation}
    R_k=\frac{\tau_c-\tau}{\tau_c}\log_2(1+\mathrm{SINR}_k),
    \label{rate}
\end{equation}
with
{\small
\begin{equation}
\mathrm{SINR}_k = 
\frac{\overbrace{p_k|\mathbb{E}[\mathbf{v}_k^H \mathbf{h}_k]|^2}^{\text{Desired term}}}
{\underbrace{\sum_{i=1}^K p_i \mathbb{E}[|\mathbf{v}_k^H \mathbf{h}_i|^2]
- p_k |\mathbb{E}[\mathbf{v}_k^H \mathbf{h}_k]|^2}_{\text{Interference term}}
+ \underbrace{\sigma^2_{\mathrm{ul}}\mathbb{E}[\|\mathbf{v}_k\|^2]}_\text{Noise term}},
\label{SINR_UATF}
\end{equation}}where $\mathbf{v}_k$ is the combining vector and $(\tau_c-\tau)/\tau_c$ is the fraction of each coherence block that is used for data transmission. We consider \ac{MR} combining, $\mathbf v_k=\mathbf{\hat h}_k$, which implies that the combiner is designed to maximize the desired term in \eqref{SINR_UATF} under \ac{MMSE} estimation, since the channel estimate is orthogonal to the estimation error. Since \ac{MR} combining avoids the matrix inversion required by \ac{MMSE} combining, which becomes increasingly costly as the numbers of users $K$ and antennas $M$ scale, we consider \ac{MR} combining for its low complexity and analytical feasibility. We acknowledge, however, that \ac{MMSE} combining is \ac{SINR}-optimal \cite{Cell_Free_Book}.

\begin{figure*}[!t]
\begin{equation}
\label{SINR_MR}
\mathrm{SINR}_k^{\mathrm{MR}} 
=\frac{\overbrace{p_k|\operatorname{tr}(\mathbf \Upsilon_k)|^2}^{\text{Desired term}}}{\sum_{i=1}^Kp_i\Big(\underbrace{\operatorname{tr}(\mathbf R_i\mathbf \Upsilon_k)}_{\text{Non-coherent interference}}+\underbrace{\eta_i\eta_k \tau^2|\boldsymbol \Gamma_{i,k}|^2|\operatorname{tr}(\mathbf R_i \boldsymbol \Psi_k^{-1}\mathbf R_k)|^2\Big)-p_k|\operatorname{tr}(\mathbf \Upsilon_k)|^2}_{\text{Coherent interference}}+\underbrace{\sigma^2_{\mathrm{ul}}\operatorname{tr}(\mathbf \Upsilon_k)}_\text{Noise term}}
\end{equation}
\hrulefill
\end{figure*}

\begin{lemma}
\label{lemma:MR_SINR}
Under \ac{MR} combining based on \ac{MMSE} channel estimation, the \ac{UATF} \ac{SINR} for user $k$ can be expressed as \eqref{SINR_MR} on top of the next page, where $\mathbf \Upsilon_k=\mathbb E[\hat{\mathbf h}_k\hat{\mathbf h}_k^H]=\eta_k\tau\mathbf R_k \boldsymbol \Psi^{-1}_k\mathbf R_k.$
\end{lemma}

\begin{IEEEproof}
The desired and noise terms in \eqref{SINR_UATF} follow directly under \ac{MR} combining by exploiting the orthogonality between the channel estimate and the estimation error, $\mathbb{E}[\mathbf{\hat h}_k^H \mathbf{h}_k]=\mathbb{E}[\mathbf{\hat h}_k^H \mathbf{\hat h}_k]=\operatorname{tr}(\eta_k\tau \mathbf R_k\boldsymbol \Psi^{-1}_k \mathbf R_k)\triangleq\operatorname{tr}(\mathbf \Upsilon_k)$. 
Furthermore, the interference term $\mathbb{E}[|\mathbf{\hat h}_k^H \mathbf{h}_i|^2]$ can be rewritten as
\begin{equation}
    \mathbb{E}[|\mathbf{\hat h}_k^H \mathbf{h}_i|^2]=\operatorname{tr}(\mathbb E[\mathbf h_i \mathbf h_i^H\hat{\mathbf h}_k \hat{\mathbf h}_k^H]).
    \label{Trace_operator}
\end{equation}
Since $\mathbf h_i$ and $\hat{\mathbf h}_k$ are zero-mean complex Gaussian vectors, we expand the fourth-order moment $\mathbb E[\mathbf h_i \mathbf h_i^H\hat{\mathbf h}_k \hat{\mathbf h}_k^H]$ using Isserlis's theorem \cite{isserlis1918formula} element-wise as
\begin{equation}
\begin{split}
&\mathbb E[\mathbf h_i \mathbf h_i^H \hat{\mathbf h}_k \hat{\mathbf h}_k^H]_{a,b}=
\sum_{c=1}^M  \mathbb E[h_i^{(a)}h_i^{*(c)}]\mathbb E[\hat h_k^{(c)}\hat h_k^{*(b)}]+\\
&\mathbb E[h_i^{(a)}\hat h_k^{(c)}]\mathbb E[h_i^{*(c)}\hat h_k^{*(b)}]+\mathbb E[h_i^{(a)}\hat h_k^{*(b)}]\mathbb E[h_i^{*(c)}\hat h_k^{(c)}].
\end{split}
\label{Element_isserlis}
\end{equation}
We notice that $\sum_{c}  \mathbb E[h_i^{(a)}h_i^{*(c)}]\mathbb E[\hat h_k^{(c)}\hat h_k^{*(b)}]$ corresponds to the matrix product $\mathbf R_i\mathbf \Upsilon_k$ and introduce the notation
$\mathbf P_{ik}=\mathbb E[\mathbf h_i\hat {\mathbf h}_k^T],$ $\mathbf C_{ik}=\mathbb E[\mathbf h_i\hat {\mathbf h}_k^H].$ With this, we rewrite \eqref{Element_isserlis} as
\begin{equation}
\mathbb E[\mathbf h_i \mathbf h_i^H \hat{\mathbf h}_k \hat{\mathbf h}_k^H]=\mathbf R_i \mathbf \Upsilon_k +\underbrace{\mathbf P_{ik}\mathbf P_{ik}^H}_{\text{Pseudo-covariance}}+\mathbf C_{ik} \operatorname{tr}(\mathbf C_{ik}^H),
\label {Compact_isserlis}
\end{equation}
where the pseudo-covariance vanishes due to circular symmetry. 
Returning to \eqref{Trace_operator} and applying the trace operator yields
\begin{equation}
\operatorname{tr}(\mathbb E[\mathbf h_i \mathbf h_i^H\hat{\mathbf h}_k \hat{\mathbf h}_k^H])=\operatorname{tr}(\mathbf R_i \mathbf \Upsilon_k)+|\operatorname{tr}(\mathbf C_{ik})|^2.
\end{equation}
To complete the proof, we note that $\mathbf h_k$ and $\mathbf h_i$ are independent if $i\neq k$, implying that $\mathbb E[\mathbf h_i \hat{\mathbf h}_k^H]$ reduces to the single term where $j=i$ in \eqref{y_pilot}, leading to 
\begin{equation}
    |\operatorname{tr}(\mathbf C_{ik})|^2=\eta_k\eta_i\tau^2|\mathbf \Gamma_{i,k}|^2|\operatorname{tr}(\mathbf R_i \mathbf \Psi_k^{-1} \mathbf R_k)|^2,
\end{equation}
which completes the proof. 
\end{IEEEproof}
One should note that pilot contamination not only affects the coherent interference in \eqref{SINR_MR}, but also shapes $\mathbf \Psi_k$, and consequently, $\boldsymbol{\Upsilon}_k$. This implies that the pilot contamination's effect on the \ac{SINR} is nontrivial, motivating a closer examination of its impact on the channel estimate.
\section{Channel Estimation with Pilot Contamination} 
To understand the effect of pilot contamination independently of the combiner choice, we now analyze how it affects the spatial properties of the channel estimate. We begin by reformulating the channel estimate's covariance matrix to explicitly separate the true covariance from the channel estimation error. For notational convenience, let  $\mathbf A_k=\sigma^2_{\mathrm{ul}}\mathbf I_M+\sum_{j\neq k} \eta_j\tau |\boldsymbol \Gamma|_{j,k}^2\mathbf R_j.
$
Since $\mathbf R_k$ is Hermitian \ac{PSD}, its square root matrix $\mathbf R_k^{1/2}$ is unique. Recalling that $\mathbf \Upsilon_k=\eta_k\tau\mathbf R_k\mathbf \Psi_k^{-1}\mathbf R_k$, we factor the $\mathbf R_k$ term in $\mathbf \Psi_k$ as $\mathbf R_k^{1/2}\mathbf I_M \mathbf R_k^{1/2}$ and apply Woodbury's identity to obtain
\begin{align}
\medmuskip=2mu\thickmuskip=3mu
&\mathbf \Upsilon_k =
\eta_k\tau\mathbf R_k(\eta_k\tau\mathbf R_k^{1/2}\mathbf I_M \mathbf R_k^{1/2} +\mathbf A_k)^{-1}\mathbf R_k \nonumber \\
&= \eta_k\tau\mathbf R_k^{1/2}\Big(\mathbf T_k-\eta_k\tau\mathbf T_k(\mathbf I_M+\eta_k\tau\mathbf T_k)^{-1}\mathbf T_k\Big)\mathbf R_k^{1/2},
\label{Rewrite}
\end{align}
where
    $\mathbf T_k\triangleq \mathbf R_k^{1/2}\mathbf A_k^{-1}\mathbf R_k^{1/2}\succeq 0.$ Note that $\mathbf T_k$ is \ac{PSD} since it is a congruence transformation of $\mathbf A_k^{-1}$, which is positive definite.
After further matrix manipulation, we obtain
\begin{align}
\medmuskip=2mu\thickmuskip=3mu
&\eta_k\tau\mathbf R_k^{1/2}\Big(\mathbf T_k-\eta_k\tau\mathbf T_k(\mathbf I_M+\eta_k\tau\mathbf T_k)^{-1}\mathbf T_k\Big)\mathbf R_k^{1/2} \nonumber \\
&= \mathbf{R}_k-\mathbf R_k^{1/2}(\mathbf I_M+\eta_k\tau\mathbf T_k)^{-1}\mathbf R_k^{1/2},
\label{Final_rewrite}
\end{align}which separates the true covariance from the suppression term induced by pilot contamination and noise.

The reformulation \eqref{Final_rewrite} enables further insights.

\begin{lemma}
For user $k$, if the pilot correlation matrices $\hat{\boldsymbol \Gamma}$ and $\check{\boldsymbol \Gamma}$ satisfy $|\hat{\boldsymbol \Gamma}_{j,k}| \geq |\check{\boldsymbol \Gamma}_{j,k}|$ for all $j$, the corresponding \ac{MMSE} channel estimates $\mathbf{ \hat h}_k$ and $\mathbf{\check h}_k$ satisfy
\begin{equation}
    \mathbf a^H\mathbb E[\mathbf{ \hat h}_k\mathbf{ \hat h}_k^H]\mathbf a\leq \mathbf a^H\mathbb E[\mathbf{ \check h}_k\mathbf {\check h}_k^H]\mathbf a,
\end{equation}
where $\mathbf a\in \mathbb C^{M\times 1}$ is an arbitrary vector, implying that increasing pilot contamination cannot increase the power of the channel estimate in any direction.
\label{Monotone_dec}
\end{lemma}

\begin{IEEEproof}
    Since $|\hat{\boldsymbol \Gamma}_{j,k}|^2 \geq |\check{\boldsymbol \Gamma}_{j,k}|^2$ for all $j$, the corresponding matrices satisfy $\hat{\mathbf A}_k\succeq \check {\mathbf A}_k$. Recalling that $\mathbf T_k\triangleq \mathbf R_k^{1/2}\mathbf A_k^{-1}\mathbf R_k^{1/2}$ and that matrix inversion reverses ordering \cite[Cor.~7.7.4]{Horn_Johnson_1985}, we obtain
           $ \hat{\mathbf T}_k\preceq\check{\mathbf T}_k$,
   and it follows that
 \begin{equation}
     (\mathbf I_M+\eta_k\tau\hat{\mathbf T}_k)^{-1}\succeq (\mathbf I_M+\eta_k\tau\check{\mathbf T}_k)^{-1}.
 \end{equation}
Since the suppression term in \eqref{Final_rewrite} is a congruence transform of $ (\mathbf I_M+\eta_k\tau\mathbf T_k)^{-1}$, increasing pilot contamination will always lead to power suppression in all directions. 
\end{IEEEproof}
Due to monotonicity, increasing the matrix $\mathbf A_k$, restricted to the column space $\mathcal R(\mathbf R_k)$, also results in larger suppression. This can be shown by formulating 
\begin{equation}
        \mathbf y^H\mathbf T\mathbf y=\underbrace{(\mathbf R_k^{1/2}\mathbf y)^H}_{\mathbf x^H}\mathbf A_k^{-1} \underbrace{\mathbf R_k^{1/2}\mathbf y}_{\mathbf x},
        \label{T_exp}
    \end{equation}
where $\mathbf y\in \mathbb C^{M\times 1}$ is an arbitrary vector and $\mathbf x \in \mathcal{R}(\mathbf R_k)$.

Beyond the overall increase in suppression, \eqref{T_exp} also reveals that suppression is inherently directional. Since $\mathbf A_k^{-1}$ has the same eigenbasis as $\mathbf A_k$ but with reciprocal eigenvalues, $\mathbf x^H\mathbf A_k^{-1}\mathbf x$ is small when $\mathbf x$ aligns with the dominant interference eigenspace of $\mathbf A_k$. This implies strong suppression in those directions. Since $\mathbf x\in \mathcal R(\mathbf R_k)$, suppression is strongest in the directions where $\mathcal R({\mathbf R_k})$ overlaps with the dominant eigenspace of $\mathbf A_k$, or equivalently, when the user's angular support overlaps with the dominant interference directions. Conversely, when $\mathcal R(\mathbf R_k)$ lies outside of the interference eigenspace, suppression is minimal, which aligns with the general idea of spatial separation mitigating the effect of pilot contamination \cite{Classic_allPC}. The interesting regime lies between the two extremes of full and no overlap. When the angular supports are close but not significantly overlapping, suppression reduces the channel estimate's power in the interferer's angular region while having a comparatively small effect on the estimate's overall power.

This spatial structure can be exploited when assigning or designing pilots. By allowing users to deliberately share or use correlated pilots, the resulting pilot contamination will steer the channel estimate away from the overlapping subspace, thereby shaping the MR combining vector to suppress interference. Since this suppression reduces the desired signal power as well, a net benefit arises only when the interference removed outweighs the desired signal sacrificed. One may notice that the coherent interference in \eqref{SINR_MR} is non-zero when pilot contamination is introduced. However, as long as $\mathbf R_i$ and $\mathbf R_k$ are sufficiently spatially separated, the coherent term is small relative to the non-coherent part. This is due to the square in $|\mathrm{tr}(\mathbf R_i\boldsymbol \Psi^{-1}_k\mathbf R_k)|^2$ making the coherent part more sensitive to spatial separation than the non-coherent interference $\mathrm{tr}(\mathbf R_i\mathbf \Upsilon_k)$. Furthermore, $\mathrm{tr}(\mathbf R_i\mathbf \Upsilon_k)$ is significantly reduced when $\mathbf \Upsilon_k$ is suppressed strongly in directions where $\mathbf R_i$ is strong. This is particularly beneficial when the interferer $i$ has a stronger channel than the desired user $k$, so that the gain from suppressing interference exceeds the loss in desired signal power. In these scenarios, using fewer pilots can achieve a better sum rate than with perfect \ac{CSI} when MR is used, as we exploit the structure of pilot contamination to reduce interference. We will demonstrate this in the next section.

\section{Numerical results}
In this section, we highlight scenarios and regimes where pilot contamination can be exploited to improve the sum rate when using MR combining. The results are based on a uniform distribution of scatterers inside the interval $[\varphi_k-\Delta,\varphi_k+\Delta]$, often referred to as the one-ring model \cite{Corr_fading}. Under this model, the $(l,m)$th element of the spatial correlation matrix \eqref{Spat_cor} for user $k$ is \begin{equation}
[\mathbf{R}_k]_{l,m}=\frac{\beta_k}{2\Delta}\int^{\Delta}_{-\Delta} e^{j2\pi (l-m)d  \sin(\varphi_k+\delta)}d \delta.
\label{Spat}
\end{equation}

We consider $K=3$ users, with user 1 and user 2 at the nominal angles $\varphi_1=0\degree$ and $\varphi_2=90 \degree$ relative to \ac{BS}, each with a \ac{SNR}, $\eta_k\beta_k/\sigma^2_{\mathrm{ul}}$, of $26$ dB. User 3 represents a user farther from the \ac{BS}, with an \ac{SNR} of $9$ dB, and its nominal angle is rotated from $0 \degree$ to $90 \degree$. 

To illustrate how pilot contamination affects the channel estimate, we show the average \ac{PAS} of the channel estimate for the rotating user (user 3) across 1000 Monte Carlo realizations in Fig.~\ref{PSD}, under full pilot sharing, i.e., $\tau=1$, and compare it with the case of perfect \ac{CSI}. The nominal angle of user 3 is set to $45 \degree$, with $\Delta=20\degree$ and $M=16$ antennas, to visualize the case where the desired user's angular support does not overlap with the interferers', while being close enough for there to be significant interference suppression. Since the averaging removes small-scale fading and channel estimation errors, the resulting spectrum isolates the spatial effect of pilot contamination, clearly showing power suppression in interfering directions and retained power in the desired angular interval. This suppression, unachievable with perfect \ac{CSI} under \ac{MR} combining, is precisely the mechanism behind the sum-rate gain, and it is especially impactful when the interferers are strong relative to the desired user.

\begin{figure}[t!]
\vspace{-7pt}
\centerline{\includegraphics[width=0.725\columnwidth]{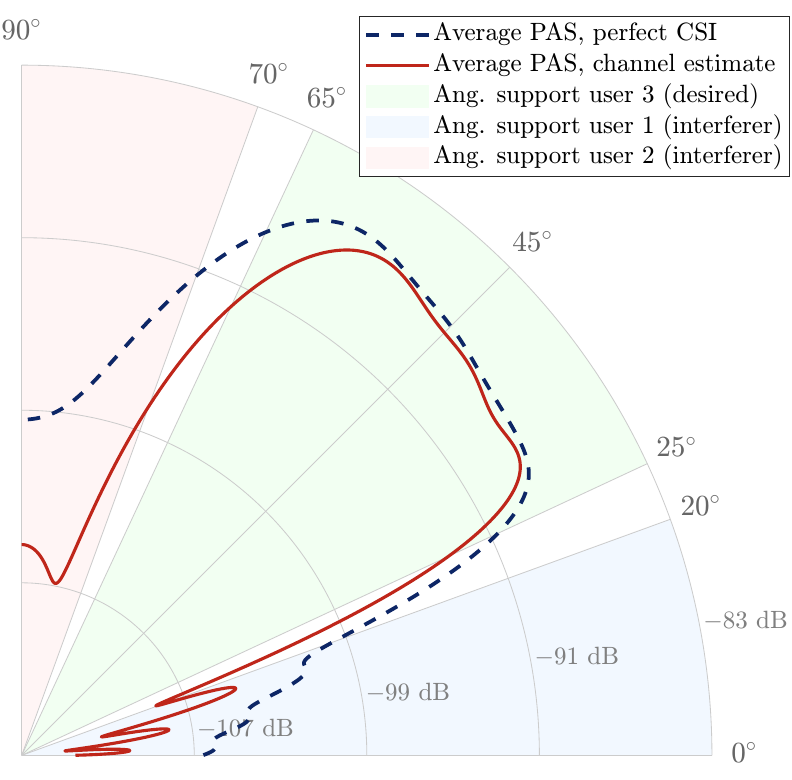}}
\caption{Average \ac{PAS}, in a polar plot, over 1000 Monte Carlo realizations of the desired user's channel estimate, with the nominal angle as $45 \degree$ and the interfering users at $0\degree$ and $90 \degree$.}
\label{PSD}
\vspace{-10pt}
\end{figure}

 In Fig.~\ref{overM}, we show the relative sum-rate gain in percent of full pilot sharing over mutually orthogonal pilots for $M=16$, $M=32$, and $M=64$, with $\Delta=25\degree$. The gain becomes positive when the nominal angle of the rotating user, $\varphi$, exceeds approximately $\Delta$ degrees offset from an interferer, indicating that pilot contamination becomes exploitable once the majority of the angular support of the desired user does not overlap with that of the interferers. To isolate the effect of the SINR on sum rate, the pre-log factor $(\tau_c-\tau)/\tau_c$ in \eqref{rate} was set to 1 in all numerical results. Comparing the relative gain across different numbers of antennas shows that the rate gain is not limited to the small-$M$ regime. To the contrary, increasing $M$ amplifies both the potential gain when pilot contamination suppresses interference and the potential loss when angular supports overlap. Note that the curves over nominal angles are not symmetric around $\varphi=45\degree$, since $\mathbf R_k$ depends on $\sin(\varphi+\delta)$, which gives the angular support a different spatial structure near $\varphi=0\degree$ and near $\varphi=90\degree$.

 \begin{figure}[t!]
 \vspace{-7pt}
\centerline{\includegraphics[width=0.9\columnwidth]{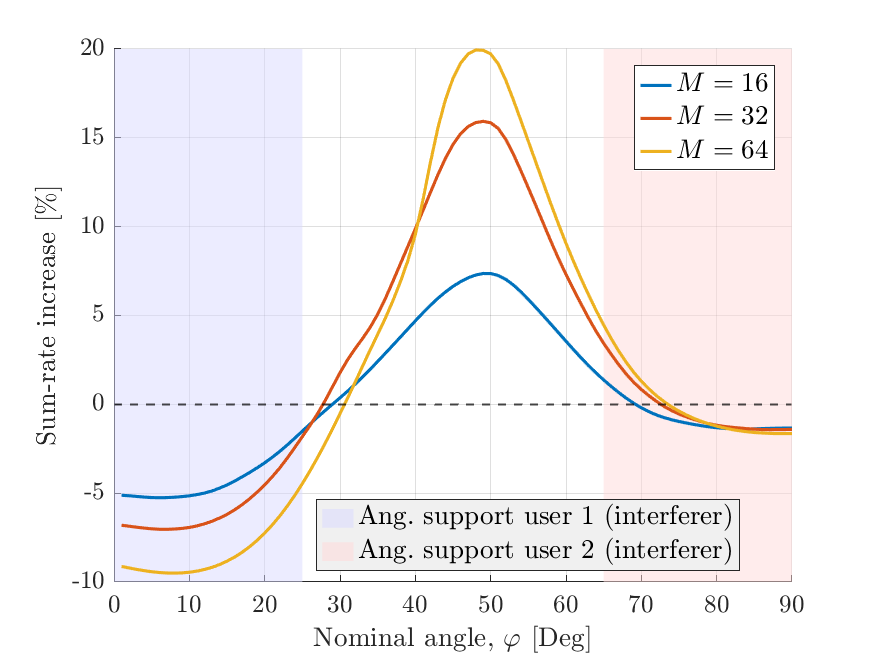}}
\caption{Sum rate increase ($\%$) of full pilot sharing over mutually orthogonal pilots for different numbers of antennas.}
\label{overM}
\vspace{-10pt}
\end{figure}

We now consider three pilot configurations: full pilot sharing ($\tau=1$), partial pilot sharing ($\tau=2$), and mutually orthogonal pilots ($\tau=3$). The pilot allocation for $\tau=2$ is determined by exhaustively searching for the combination that maximizes the sum rate, subject to the constraint that both pilots have to be used to distinguish it from $\tau=1$. Fig.~\ref{colorMap} shows the pilot configuration that achieved the largest sum rate for different $\Delta$ and nominal angles $\varphi$ with $M=64$. In areas where leakage can be significantly suppressed towards user 1 and user 2, it is clear that full pilot sharing outperforms the other configurations. However, when the rotating user is spatially close to only one of the interferers, with little angular overlap, sharing a pilot with just that user is optimal. In regions where the angular support significantly overlaps with an interferer, mutually orthogonal pilots achieve the largest sum rate. We further note that for small $\Delta$ and a nominal angle $\varphi$ far from both interferers, the users are nearly orthogonal in spatial structure, yielding only a marginal gain from pilot sharing.

\begin{figure}[t!]
\centerline{\includegraphics[width=0.9\columnwidth]{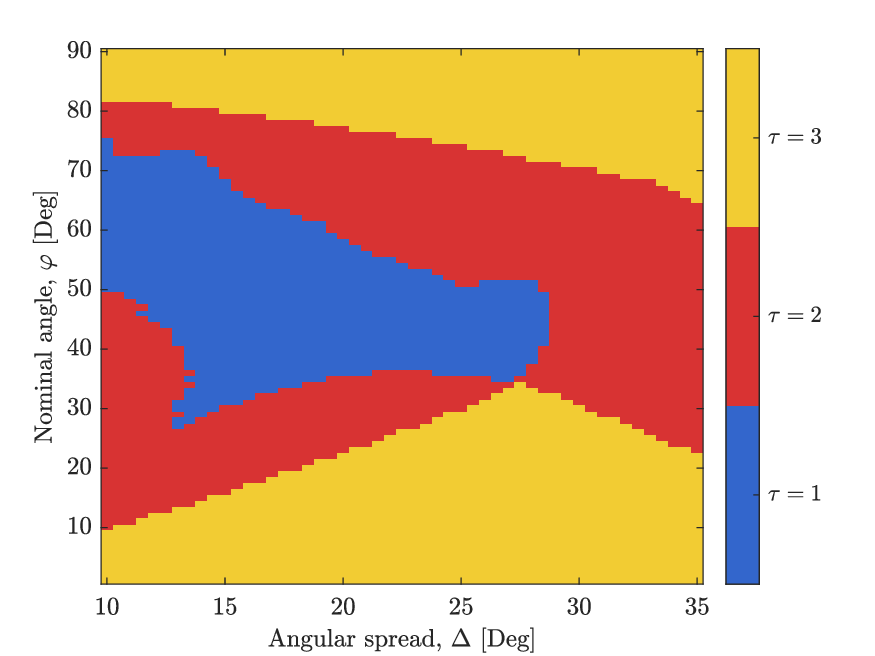}}
\caption{Regions where full pilot sharing ($\tau=1$), partial pilot sharing ($\tau=2$) and mutually orthogonal pilots ($\tau=3$) achieved the largest sum rate.}
\label{colorMap}
\vspace{-10pt}
\end{figure}

Having established that the optimal number of pilots depends on the users' angular supports, we now illustrate this numerically for the specific example with $\Delta=20 \degree$ and $M=16$. In Fig.~\ref{fig:Analytical_opt}, the sum rate is shown over nominal angles with full pilot sharing, partial pilot sharing, and mutually orthogonal pilots. To further validate the results, we show the sum rate with full and partial pilot sharing, averaged over 10 000 Monte Carlo channel realizations, and compare it with perfect \ac{CSI} under \ac{MR} and normalized \ac{MR} (N\ac{MR}) combining in Fig.~\ref{fig:Numerical_opt}. Note that N\ac{MR} combining refers to normalizing the \ac{MR} combiner to unit norm. The labeled regions indicate the pilot configurations that achieved the largest sum rate. Full pilot sharing achieves the largest sum rate when the rotating user is close enough to both user 1 at $0 \degree$ and user 2 at $90\degree$ to suppress leakage in both directions, whereas mutually orthogonal pilots are preferred under strong angular overlap. Between these regimes, partial pilot sharing is optimal. Additionally, in areas where mutually orthogonal pilots are ideal, partial pilot sharing still performs nearly as well, since the rotating user can avoid sharing a pilot with the interfering user that has overlapping angular support, opting instead to share with the user with near-orthogonal spatial correlation. These results numerically demonstrate the different pilot-sharing regimes and how much deliberate pilot contamination can outperform mutually orthogonal pilots or perfect \ac{CSI} in terms of sum rate under \ac{MR} or N\ac{MR} combining.
\begin{figure}[t!]
\vspace{-7pt}
\centering
\subfloat[Sum rate under \ac{MR} combining.\label{fig:Analytical_opt}]{%
  \includegraphics[width=0.9\columnwidth]{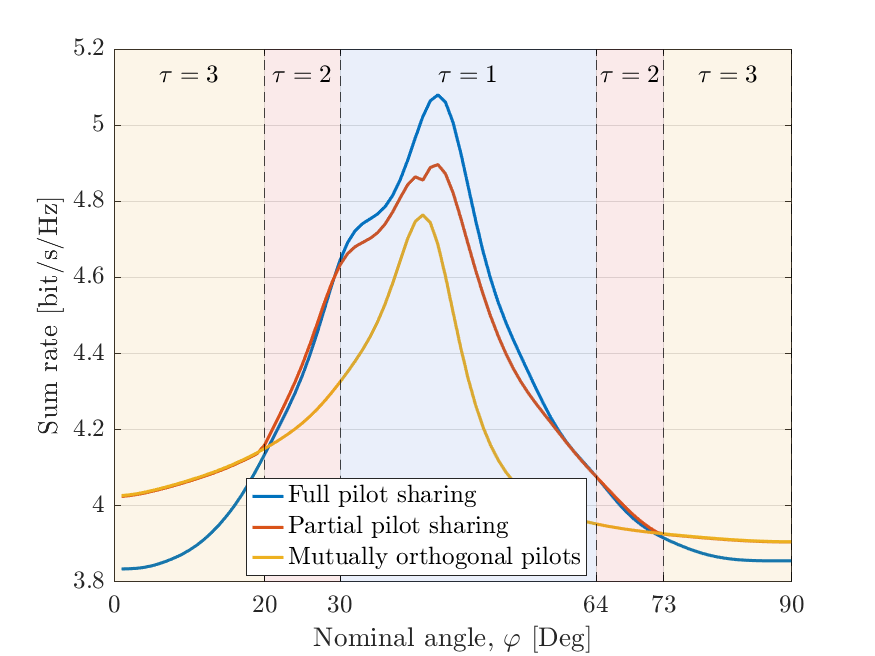}}\\
\vspace{-7pt}
\subfloat[Averaged sum rate over 10 000 Monte Carlo realizations under \ac{MR} and NMR combining.\label{fig:Numerical_opt}]{%
  \includegraphics[width=0.9\columnwidth]{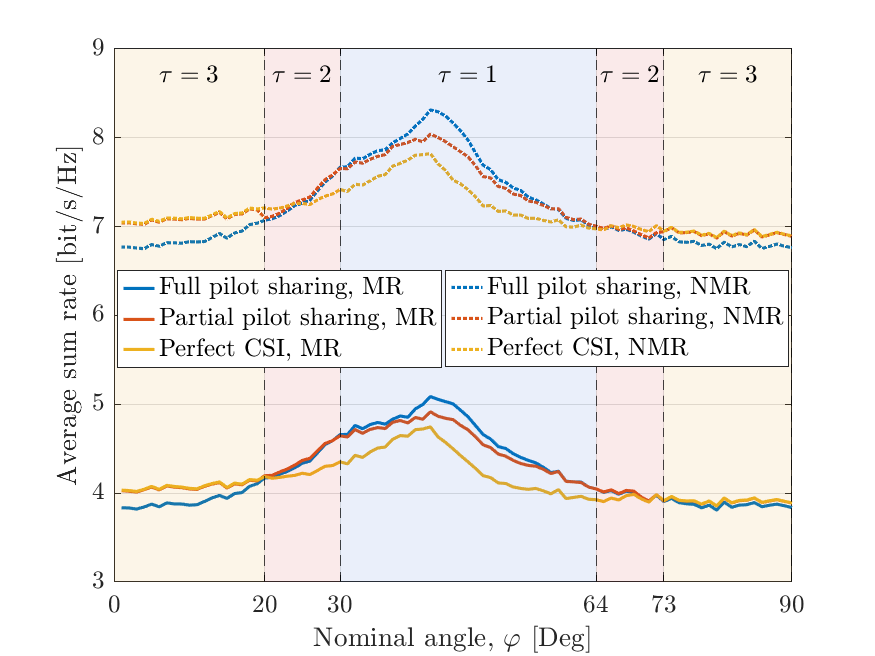}}
\caption{Sum rate over nominal angles for the rotating user, with $M=16$ antennas. Shaded regions indicate the pilot configuration that achieved the largest sum rate.}
\label{fig:opt}
\vspace{-6pt}
\end{figure}
\section{Conclusion}
In this work, we have shown that introducing pilot contamination can increase \ac{UL} sum rate under \ac{MR} combining by shaping \ac{MMSE} channel estimates, thereby jointly reducing pilot overhead and improving performance. In addition, we have derived a general \ac{UATF} \ac{SINR} bound for arbitrary pilots under \ac{MR} combining. We have highlighted that the key mechanism behind the sum rate gain is the directional power suppression of channel estimates, which can significantly reduce interference under \ac{MR} combining. Numerical results demonstrate the regimes where deliberate pilot contamination outperforms mutually orthogonal pilots and perfect \ac{CSI} in terms of sum rate, and show that the gain increases with the number of antennas. In more general scenarios with more users, at least some users will likely have angular support that enables this exploitation, suggesting that understanding how pilot contamination can be leveraged in this way can help reduce pilot overhead. Future work includes extending the framework to power allocation and MMSE combining.
\bibliographystyle{IEEEtran}
\bibliography{IEEEabrv,references}
\end{document}